\documentclass [twocolumn,natbib,envcountsame] {svjour3}

\usepackage[hyphens]{url}
\usepackage{microtype}
\usepackage{blkarray}

\usepackage[utf8]{inputenc}
\usepackage{amssymb,amsmath}
\usepackage{mathrsfs}
\usepackage{graphicx}

 \usepackage{color,hcolor}
 \usepackage[colorlinks=false]{hyperref}
 \usepackage{mathptmx}
\usepackage{inconsolata}
\DeclareSymbolFont{newsymb}{T1}{fourier-mlit}{m}{it}
\DeclareFontFamily{T1}{fourier-mlit}{}
\DeclareFontShape{T1}{fourier-mlit}{m}{it}{<-> fourier-mlit}{}
\DeclareMathSymbol{\varrho}{\mathalpha}{newsymb}{"25}

\DeclareSymbolFont{cmsy}{OMS}{cmsy}{m}{n}
\DeclareSymbolFontAlphabet{\mathcalus}{cmsy}

\newcommand\bket[1]{\ket{\normalfont\texttt{#1}}}
\newcommand\bbra[1]{\bra{\normalfont\texttt{#1}}}
\newcommand\ca{\mathcalus{C}_A}
\newcommand\ttrm[1]{\normalfont\texttt{#1}}

\newcommand{\tr}{\operatorname{tr}}
\newcommand{\ket}[1]{|{#1} \rangle}
\newcommand{\bra}[1]{{\langle {#1}|}}
\newcommand{\braket}[2]{\langle {#1} | {#2} \rangle}

\renewcommand{\phi}{\varphi}

\newcommand\PT[0]{(\mathrm{T}\otimes\mathrm{I})}

\def\mh{\mathcal{H}}
\def\R{}%#1{{\textcolor{red}{#1}}}
\def\openone{\mathrm{I}}
\smartqed

\begin{document}

\title{Extrapolated Quantum States, Void States, and a Huge Novel
Class of Distillable Entangled States}

\titlerunning{Extrapolated Quantum States, Void States, and a Huge Novel
Class of Distillable Entangled States}

\author{Michel Boyer
\and Aharon Brodutch
\and Tal Mor }
\authorrunning{M. Boyer,  A. Brodutch and T. Mor  }
\institute{M. Boyer \at
DIRO, Universit\'e de Montr\'eal, Canada,\\ \email{boyer@iro.umontreal.ca}
\and A. Brodutch \at Department of Physics \& Astronomy and Institute for Quantum Computing, University of Waterloo, Canada,\\\email{aharon.brodutch@uwaterloo.ca}
  \and T. Mor \at Computer Science Department, Technion, Israel,\\ \email{talmo@cs.technion.ac.il}
}

\maketitle

\setcounter{footnote}{0}   

\begin{abstract}
A nice and interesting property of any 
pure tensor-product state is that 
each such state
has distillable entangled states
at an arbitrarily small distance
$\epsilon$ in its neighbourhood.
We say that such nearby states are $\epsilon$-entangled,
and we call the tensor product state in that case, 
a ``boundary separable state'', as there is entanglement 
at any distance from this ``boundary''.
Here we find a huge class of separable states that also 
share that property mentioned above -- they
all have  
$\epsilon$-entangled states 
at any small distance in their 
neighbourhood.  
Furthermore, the entanglement they have is proven to be 
distillable.  We then extend this result to the discordant/classical cut and show that all classical states  (correlated and uncorrelated)  have discordant states at distance $\epsilon$, and provide a constructive method for finding $\epsilon$-discordant states. 

\end{abstract}

\section{Introduction}

Studying the structure of the set of quantum states has been a central topic in quantum information research \citep{geometrybook}. Particular emphasis is on   the tensor product structure  of this space in terms of entanglement \citep{EntanglementRMP} and general correlations \citep{DiscordRMP,BellRMP,Groisman2007}. The study has lead to the identification of important families of states such as  Werner states~\citep{Werner},
bound entangled states~\citep{Horodecki1997333,UPB}, and the 
W-states~\citep{PhysRevA.62.062314}, as well as interesting sets of bases such as unextendable  product bases (UPB) \citep{UPB} and locally indistinguishable bases \citep{NLWE}. 

One method which has been particularly powerful in the study of state-space is interpolation, i.e.studying the states that lay between two known states with different properties. Interpolation has been used in the study of robustness to various types of noise \citep{vidal} and learning about the  ball of separable state \citep{ball}. The complementary method, extrapolation, has also been used in some cases for example  in the study of  non-signaling theories \citep{nonsig} where trace one Hermitian operators with negative eigenvalues were required for  larger than quantum violations of Bell inequalities. 

Here we use both extrapolation and interpolation\footnote{A preliminary version of this work (without discordant states) appeared in TPNC-2014  \citep{TPNC-2014}.}
 to study the boundaries between various subsets of quantum spaces with particular emphasis on boundary separable states~- separable states that are arbitrarily close to an entangled state cf. Definition \ref{Def:BS}.

Any pure tensor product state has entangled states
near it, at any distance (i.e.  arbitrarily close), making it boundary separable (cf. corollary \ref{Cor:pureBS}) .  Another simple example of a boundary separable state is a Werner-state
$\lambda/3 [\rho_{\psi_+}+\rho_{\phi_+}+\rho_{\phi_-}] 
+ (1-\lambda) [\rho_{\psi_-}]$ (built from the four Bell states) 
with $\lambda = 1/2$ which has 
entangled states near it, at any distance. 

Is the property of being separable yet having entangled states nearby at any
distance common? Or is it rare?
Furthermore, what can we learn about the type of entanglement that those nearby entangled
states have? For two qubits, it is known~\citep{HorodeckiPRL97} that the entanglement is always distillable.
For qudits, cf proposition~\ref{prop:real}.

Our main result is  a number of families of boundary separable quantum states. We show that states in these families are arbitrarily close to distillable entangled states and give a constructive method for identifying these states.  
We also provide similar results for discord (cf. Corollary \ref{cor:BC} for example), showing that all classical  states (classically correlated and uncorrelated)   are boundary classical (in fact there is a discordant and thus non classical state arbitrarily close)  and providing a constructive way to find an arbitrarily close discordant state.  Our results are presented in order of complexity starting from two qubit examples and continuing  to more general results involving qutrits, qudits and multi-qubit systems. \R{ In most cases the results are presented through examples. The general implications are discussed at the end of each section.}

\subsection{The set of quantum states}\label{sec:set}
Given a Hilbert space $\mh$, the set of quantum states (i.e.the set of positive semidefinite trace one  Hermitian operators) on $\mh$  is convex. If the Hilbert space has a physically meaningful tensor product structure  $\mh=\mh_A\otimes\mh_B$ it is useful  and interesting to consider subsets of states based on this structure. These subsets are often not convex and are hard to characterize. In the work presented here we will use convex and affine mixtures of states to study the states that lay at the boundary of these  subsets.  

One important subset of states  is the set of pure states $\ket{\Psi}\bra{\Psi}$ (i.e.states of rank 1). In this set we identify  a smaller subset of   pure product states of the form $\ket{\Psi}\bra{\Psi}=\ket{\psi}\bra{\psi}\otimes \ket{\phi}\bra{\phi}$.  All pure states that are not product are called \emph{entangled pure states}. Pure states are extremal points  in the set of all states. 
% pure states have a unique decomposition into a convex sum. 

There are various ways to similarly divide  the set of all states. One division  is into  the complementary  subsets:  separable states and entangled states. For any 
% local Hilbert space 
bipartite system whose Hilbert spaces are of
dimension at least $2$, both sets, entangled and separable,    have a finite volume in the set of all states \citep{EntanglementRMP}. The set of entangled states is not convex and in general it is hard to identify whether a state is entangled or separable
\citep{Gurvits:2003:CDC:780542.780545}.  
At the border between separable and entangled states are the boundary separable and $\epsilon$- entangled states.  Some properties of this  boundary were  previously studied in relation to non-linear entanglement witnesses \citep{witnesses} where it was shown that the set of all separable states is not a polytope (see also  \citet{EntanglementRMP}).  Our main results are specific families of boundary separable and $\epsilon$-entangled states; some of these families such as those close to thermal states are of particular importance in quantum computing. 
\R{There are a number of physically meaningful ways to  divide the set of entangled states. One that we will use here is to divide the set into two disjoint subsets, distillable and bound-entangled states. A state $\rho$ is distillable if it is possible to distill many copies of $\rho$ into a maximally entangled state.  Clearly separable states cannot be distilled, but surprisingly there are  entangled states that cannot be distilled. These are known as \emph{bound entangled} states. It is known that states with a non-positive partial transpose are bound entangled, similarly if the total dimension of the Hilbert space is not larger than 6, all entangled states are distillable (and have non-positive partial transpose) \citep{EntanglementRMP}}.
\begin{figure}
\includegraphics[width=\columnwidth]{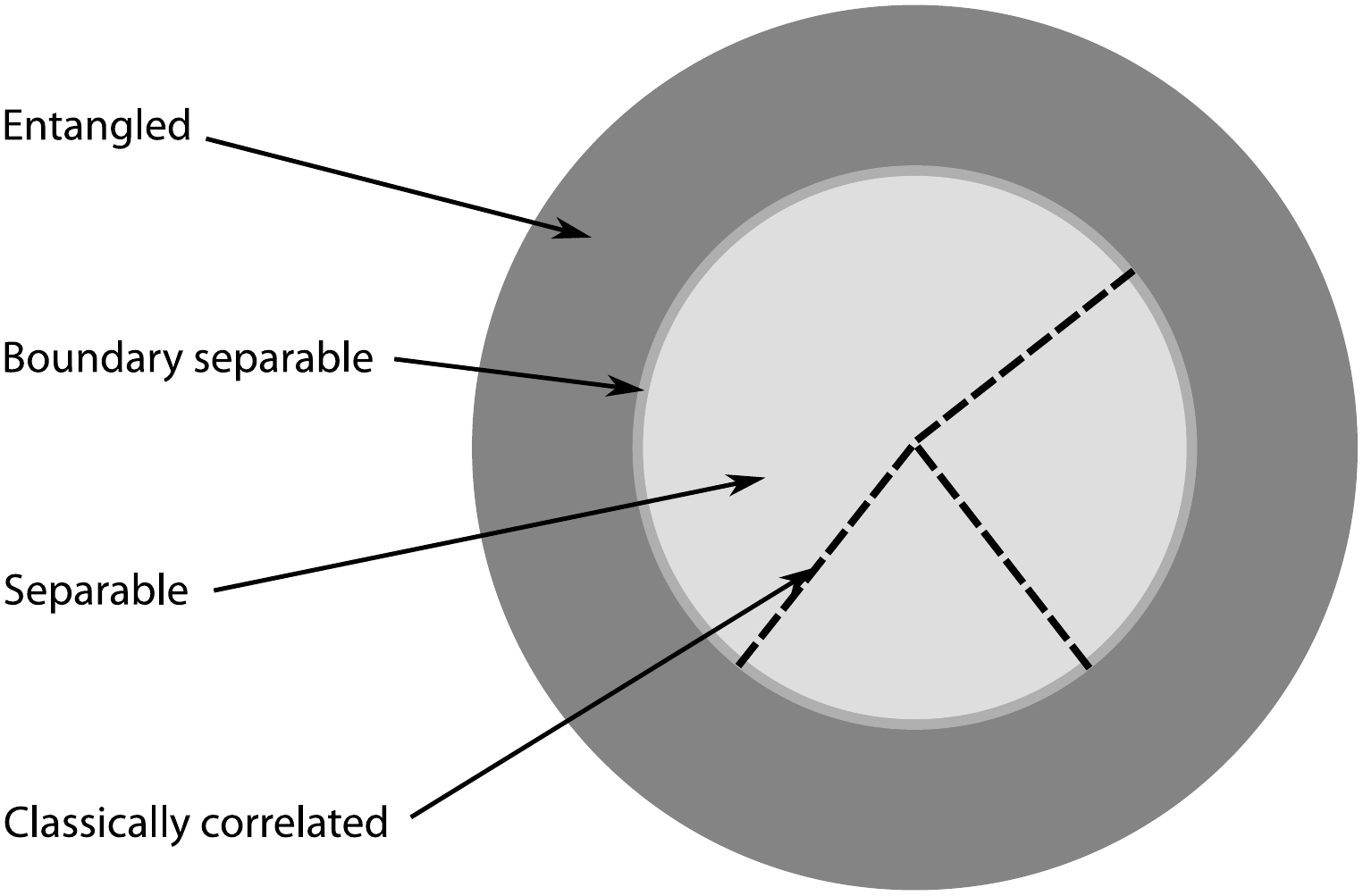}
\caption{{\bf The set of quantum states} is convex. A  convex subset of those states are separable and its (non-convex) complement is the set of entangled state, at the boundary are the boundary separable  states. Separable states are mostly dissonant (discordant and separable), the remainder are classical 
states with respect to $A$ ( correlated or uncorrelated). All classical states are boundary classical, but some are also boundary separable.    \label{fig:states}}
\end{figure}

A different classification of the set of all states is into the complimentary subsets: Discordant states  and states that are classical with respect to A \citep{DiscordRMP}, the latter are sometimes called classical-quantum (see Sec. \ref{sec:discord} for precise definitions). For simplicity we  use $\ca$ to denote the set of states that are classical with respect to A and note that the set of classical states, i.e.those that are classical with respect to both A and B, are a subset of $\ca$. The set of discordant states is $\overline{{\ca}}$. The classification into $\ca$ and $\overline\ca$  shares some properties with the classification of pure states. For example, like the set of pure product states which is vanishingly small in the set of all pure states (i.e.it requires strictly fewer parameters to characterize a pure product state than to characterize a generic  pure state), the set $\ca$    
 is vanishingly small in the set of all  states \citep{classicalvolume}. Moreover, for pure states,  discord and entanglement coincide. However, in general we can only say that entangled states are always discordant and classical  states (with respect to $A$, $B$ or both) are always separable. There is,  an intermediate regime of discordant-separable or \emph{dissonant} states \citep{unified} (see fig. \ref{fig:states}).  As we will show below, all classical states (with respect to $A$, $B$ or both) are also boundary classical, i.e.  a state is either discordant or there is a discordant state arbitrarily close to it.  

Although discord refers to a specific quantity  \citep{OZ}  other similar quantities exist  \citep{DiscordRMP}. Like entanglement monotones each measure has its own domain but generally there is an unambiguous way to quantify states as uncorrelated (product), 
classically correlated, dissonant (discordant and separable) and entangled \citep{DiscordRMP}. A  caveat on the last statement  is that  in general discord-like  measures  are not necessarily defined to be symmetric with respect to the parties involved and the cut between discordant and classical depends on this choice.  Here we mostly consider the asymmetric versions that were discussed  in the early literature, in particular the sets $\ca$ and $\overline\ca$; however all of our results apply to the various symmetric versions of discord in \citet{DiscordRMP}. 

A final classification is in terms of the eigenvectors of the state. A state has a product eigenbasis\footnote{The term product basis  should not be  confused with the   \emph{classical basis} of  \citep{Groisman2007} which is a product of local bases, rather than a basis of product states.} if it can be diagonalized in an orthonormal basis of product states. Surprisingly the decomposition of a state into eigenstates is not sufficient to tell us about entanglement  or discord. In general non-degenerate separable  states can have a non-separable  eigenbasis \citep{Groisman2007}.   On the other hand a discordant state can have a product eigenbasis. 

\section{Notations and Terminology}

\subsection{\R{General notation}}
In the majority of cases below we consider bipartite states as operators  on a Hilbert space  $\mathcal{H}_A\otimes\mathcal{H}_B$ (in Sec. \ref{sec:multi} we also consider mulitipartite systems). We use $\mathrm{T}$ to denote the transpose map, similarly $\PT$ denotes partial transposition on $A$. Distance between two states $\rho$ and $\tau$ will be given by the trace distance $\delta(\rho,\tau)=\frac{1}{2}\mathrm{Tr}|\rho-\tau|$.

\subsection{Entanglement and separability}
A convex mixture of pure states is called mixed. A quantum state $\rho$ on $\mathcal{H}_A\otimes\mathcal{H}_B$ is called a product state if it can be decomposed into $\rho=\rho_A\otimes\rho_B$, where $\rho_A$ and $\rho_B$ are states on $\mh_A$ and $\mh_B$ respectively; $\rho$ is called separable  if it can be decomposed into a convex sum 
of product states, otherwise it is called entangled (cf Appendix \ref{app:Peres}). Product states are also separable.

\subsection{Boundary Separable States and $\epsilon$-Entangled States}

\begin{definition}\label{Def:BS}A \emph{boundary separable state} 
is a separable state  $\rho_b$
such that for \emph{any} $\epsilon > 0$, 
there is an entangled state $\rho_e$ for
which $\delta(\rho_b,\rho_e) \leq \epsilon$,  
i.e. there are entangled states arbitrarily close to $\rho_b$. \end{definition}

Notice
that for any density operator $\rho$, and $0 \leq \epsilon \leq 1$, if
\begin{equation}\label{epsilonneighbour}
\quad\tau_\epsilon = (1-\epsilon)\rho_b + \epsilon \rho
\end{equation}
then 
$\delta(\tau_\epsilon,\rho_b) 
= \frac{\epsilon}{2}\mathrm{Tr}|\rho-\rho_b| 
= \epsilon\delta(\rho,\rho_b)$
and thus
\[
\quad \delta(\tau_\epsilon,\rho_b)\leq \epsilon \ .
\]
The trace distance between $\tau_\epsilon$ given by \eqref{epsilonneighbour} 
and the (boundary) separable state $\rho_b$ is at most $\epsilon$ 
but it may be much smaller than $\epsilon$;
 it is $\epsilon$ iff $\delta(\rho,\rho_b) = 1$ 
i.e. if $\rho_b$ and $\rho$ are orthogonal (have orthogonal support). 
 
\begin{definition}An $\epsilon$-\emph{entangled} state is an  entangled state $\rho_e$ such that there is 
a boundary separable state $\rho_b$ for which
$\delta(\rho_e,\rho_b)\leq \epsilon$; i.e.
it is at trace distance at most $\epsilon$ from
a boundary separable state. \end{definition}
As an example, the Werner state with $\lambda=1/2$ is a boundary separable state 
and mixing it with $\rho_{\psi_-}$ gives $\epsilon$-entangled states.

There are separable states $\rho_b$ for which there exists a state $\rho$ 
such that all the states $\tau_\epsilon$ given
by \eqref{epsilonneighbour}
are entangled for $\epsilon$ small enough, $\epsilon \neq 0$.
There is  a continuous path starting from $\rho_b$
and going straight in the direction of $\rho$ 
whose initial section contains only $\epsilon$-entangled states.
Note that for $\epsilon=0$ the resulting
state $\tau_0$ is the boundary separable-state $\rho_b$ itself; $\tau_0 = \rho_b$.
As an example, again,
the Werner state with $\lambda=1/2$ is a boundary separable state,
such that mixing it with $\rho_{\psi_-}$ as in \eqref{epsilonneighbour}
gives epsilon-entangled states, and there is a continuous path from this Werner
state and all the way to the fully entangled state $\rho_{\psi_-}$.

\subsection{``Extrapolated States'' and ``Void States''}
Given any two states $\rho_0$ and $\rho_1$, the operators 
$\rho_t = (1-t)\rho_0 + t\rho_1$ are clearly always Hermitian with trace $1$;
when $0\leq t \leq 1$, they are (mixed) states, all on a straight line segment between $\rho_0$ and $\rho_1$;
those mixed states are obtained by \emph{interpolation} (convex combination) of two states.
Let us now introduce three additional definitions:
\newcounter{itno}
\begin{list}{\upshape\alph{itno})}
{\usecounter{itno}}
\item 
 When $t<0$,  $\rho_t$ is on the same straight line
but is no longer between $\rho_0$ and $\rho_1$; 
in general, if $\rho_0 \neq \rho_1$ and all the eigenvalues 
of $\rho_0$ are strictly positive, then there are values 
of $t<0$ such that $\rho_t$ is a state;
we call such states \emph{extrapolated states}. 

Note that if $\rho_0 = \ket{0}\bra{0}$ and 
$\rho_1 = \ket{1}\bra{1}$, then
$(1-t)\rho_0 + t\rho_1  = (1-t)\ket{0}\bra{0} + t\ket{1}\bra{1}$ is 
not a state (it is not  positive semi definite) as soon as $t<0$ 
(or $t>1$). 

There may be some value $m<0$ such that 
$\rho_t$ is no longer positive semi-definite for $t<m$,
thus no longer a state (hence it is not a physical entity), 
while it is still positive semi-definite for $t=m$. 

The condition that the eigenvalues of $\rho_0$ be all positive
is sufficient for defining extrapolated states,
but not necessary. 
One can extrapolate carefully-chosen states that 
have some $0$ eigenvalues. 
Extrapolation somewhat behaves like subtraction: 
if $t<0$, then $\rho_t = (1+|t|)\rho_0 - |t|\rho_1$. 
We will be interested only with extrapolations with $t < 0$ 
though $t>1$ could also provide extrapolations.
\item
A \emph{void state} 
is a quantum state that has exactly one zero eigenvalue.
Namely, when diagonalized, it has exactly one zero on the diagonal.
\item
A \emph{$k$-void state} (of dimension $N>k$)
is a quantum state that has exactly $k$ zero eigenvalues\footnote{Note that a separable $N-1$-void state is a tensor product state.} \R{(similarly, it has rank $N-k$)}.
\end{list}

\subsection{Discord and classical correlation}\label{sec:discord}

We consider a  state $\rho$ of a bipartite system $AB$ with marginals $\rho_A$ and $\rho_B$.  

\begin{proposition}\label{DiscordCaract}Let $\rho$ be a state of $\mathcal{H}_A\otimes\mathcal{H}_B$ and $\big\{\ket{i}\big\}$ be a basis of
$\mathcal{H}_A$. Then the following three statements are equivalent \citep{DiscordRMP}. 
\begin{enumerate}
\item There is a set of states  $\{\tau_i\}$ on $\mh_B$ such that 
\[\quad\rho=\sum\lambda_i\ket{i}\bra{i}\otimes\tau_i,\]
with $\lambda_i\ge0$ and $\sum_i \lambda_i =1.$

\item There is a set of unitary operators $U_i$ such that 
\[\quad\rho=\sum_{i,j}\mu_{ij}\ket{i}\bra{i}\otimes U_i\ket{j}\bra{j}U_i^\dagger,\]
where  $\mu_{ij}\ge0$ are the eigenvalues of $\rho$.

\item $\rho$ is invariant under the action of the local  dephasing channel  $D$ defined by
$D\Big(\ket{i_1}\bra{i_2}\otimes \tau \Big) = \delta_{i_1i_2}\, \ket{i_1}\bra{i_2} \otimes \tau$ on the basis $\big\{\ket{i}\big\}$:
%
% (\cdot)=\sum_i\ket{i}\bra{i}\cdot \ket{i}\bra{i}$.
\[\quad D(\rho)=\rho.\] \label{itfour}
 \end{enumerate}
\label{prop:discord}
\end{proposition}
\begin{definition}The state $\rho$ is said to be \emph{classical with respect to the basis $\big\{\ket{i}\big\}$} of $\mathcal{H}_A$ if it satisfies one of the above conditions.
\end{definition}
These conditions imply that the decomposition  $\rho_A=\sum_i\lambda_i\ket{i}\bra{i}$ with all $\lambda_i\ge0$ is a necessary condition for  $\rho$ to be classical  in the basis $\{\ket{i}\}$ \citep{BTdemons}.  

\begin{definition}
A state $\rho$ is said to be \emph{classical with respect to $A$}  if there is a basis of $\mathcal{H}_A$ with respect to which it is classical;
the set of classical states with respect to $A$ is denoted $\mathcalus{C}_A$.
A state $\rho$ which is not in $\mathcalus{C}_A$  is called \emph{discordant} \citep{DiscordRMP,HV,OZ}  \footnote{The term classical is used in a variety of ways in the literature, here we use it in the sense of correlations as in \citep{DiscordRMP}}. The set of discordant states is $\overline\ca$.

\end{definition} 
\begin{remark}
It is important to notice that all classical states (i.e.classical with respect to both $A$ and $B$) are in $\mathcalus{C}_A$ and thus that a state not in $\mathcalus{C}_A$ (i.e. discordant) cannot be classical.
We will use that fact to build non classical (in fact, discordant) states arbitrarily close to any classical state.
\end{remark}
\begin{remark}
 Any state which is not product is called \emph{correlated}. The set $\ca$ contains both correlated and uncorrelated states while all states in  $\overline\ca$ are correlated. These are sometimes called \emph{quantum correlated} \citep{DiscordRMP}.   
\end{remark}

% We will sometimes use the term \emph{classically correlated} to refer to a classical state which is not a product state. 

When  $\rho_A=\sum_i\lambda_i\ket{i}\bra{i}$ is non degenerate, i.e., $\lambda_i\ne\lambda_j$ if $i\neq j$, the conditions above provide a very simple method to check if $\rho$ is in $\mathcalus{C}_A$ \citep{BTdemons}. When $\rho_A$ is degenerate one has to check over all its possible eigenbases. 

\subsubsection{Boundary classical states}

In the same way as above it is possible to define boundary classical states and $\epsilon$-discordant states. As we will see this definition is superfluous since all 
classical states are also boundary classical; cf. Corollary \ref{cor:BC}. 

\begin{figure}
{\includegraphics[width=\columnwidth]{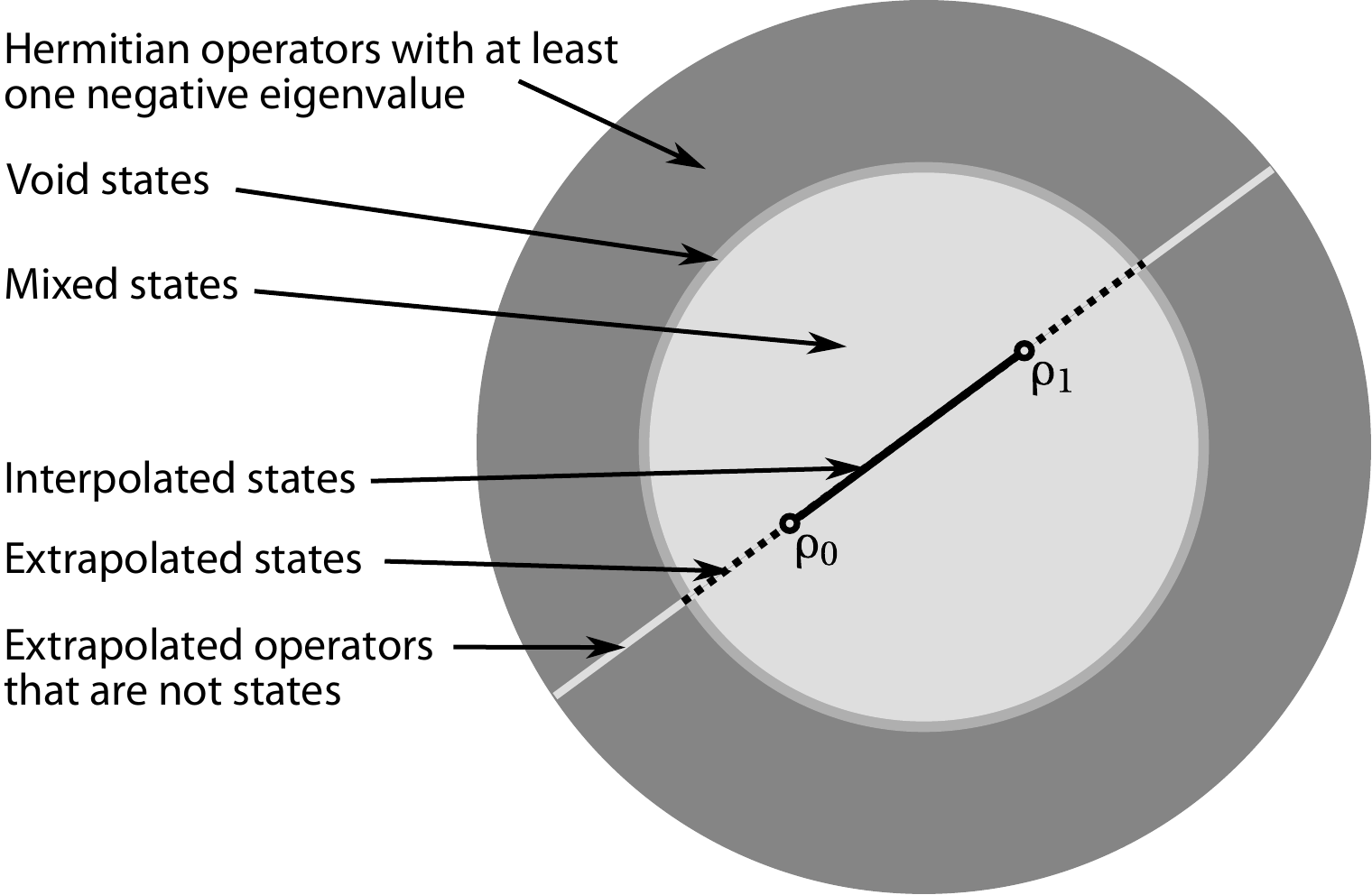}}

\caption{{\bf The interpolated and extrapolated Trace 1 Hermitian operators  $(1-t)\rho_0+t\rho_1$}. For $0 \le t \le 1$ the operator is an (interpolated) state, i.e. it is a convex mixture of states and  and we are guaranteed it is a state. For $t<0$ or $t>1$ the state is extrapolated. It will continue to be a state until one of the eigenvalues becomes negative. The states on the boundary are void states, they have at least one vanishing  eigenvalue.     \label{fig:ext}}
\end{figure}
\section{Two Qubits}

Our first example of 
2-party boundary separable-states (and the derived 
$\epsilon$-entangled states)
is obtained  
by starting from a completely mixed state
and ``fully subtracting'' one of the eigenstates, to obtain
a separable void state.
Our second example uses a different---yet very interesting state
to start with---the thermal state. As in the first example,
a void-state is generated from the thermal state (via extrapolation)
by subtracting one of the eigenstates.
Our third example uses a 2-void state
instead of a simple (1-)void state (and we also discuss
here the case of 3-void state which in this case is a tensor product state).
Our last two 2-qubit examples provide generalizations to less
trivial cases including discordant states and separable states without a product eigenbasis.  Since two qubit states that are entangled are all distillable~\citep{HorodeckiPRL97}, the states           
obtained are thus also distillable.

\subsection{Example 1 -- The Extrapolated Pseudo-Pure State of Two Qubits}\label{section:pseudo-pure}

Mixing a completely mixed state $\rho_0$ with an arbitrary state $\rho_1$ to yield the pseudo pure state (PPS) $\rho = (1-t) \rho_0 + t \rho_1$ is found to be extremely useful in quantum information processing (e.g. in NMR quantum computing). To the best of our knowledge, an extrapolated state of the form 
$\rho = (1+|t|) \rho_0 - |t| \rho_1$ was never used.
This ``extrapolated pseudo pure state" (EPPS), whenever it is a legal quantum state, shares with the conventional PPS the fact that applying any unitary transformation acts only on $\rho_1$.

An  interesting special case of this EPPS is when $|t|$ is exactly sufficiently large to make one eigenvalue vanish
% disappear 
(become zero). 
If $\rho_1$ is a pure tensor product state, then the resulting $\rho$ is a void state. We assume here that the 
subtracted tensor product state is written in the computational 
basis, e.g., it is $\bket{11}\bbra{11}$ and $m=t=-1/3$.

\begin{proposition}\label{prop-trom-zero}
If the standard basis is the eigenbasis of a state $\rho$ on $\mathscr{H}_2\otimes\mathscr{H}_2$, and if the eigenvalue of 
$\bket{11}$ is $0$, and the other
three eigenvalues are $1/3$, 
then there are states arbitrarily close 
to $\rho$ that are entangled. [The same holds,
with obvious adjustments, for any
other tensor-product eigenstate that has a zero eigenvalue.]
\end{proposition}

We avoid proving this proposition as we
later (in example 4) prove a more general result,
containing the above (and also example 2) as special cases. 
The above
mentioned 
(very basic) 
example is mainly given 
for historical reasons, as it was the first example 
we found.

For $j$ fixed, let
\[\quad
\rho = \frac{4}{3} \left[\frac{1}{4}\sum_{i=0}^3 \ket{i}\bra{i}\right] - \frac{1}{3} \ket{j}\bra{j}  
 = \frac{1}{3} \sum_{i=0;i\ne j}^3 \ket{i}\bra{i}.
 \]
This is obtained   
by choosing $|j\rangle$ (viewed as a two bit integer from $0=\ttrm{00}_2$ to $3=\ttrm{11}_2$) to be any product state 
$j\equiv j_{AB} = j_A\otimes j_B$,
where the two parties are $A$ for Alice's qubit 
and $B$ for Bob's. In fact, for all values of $t$ between $0$ and $-1/3$, the Hermitian operators
\[\quad
\rho_t =  (1-t)\left[\frac{1}{4}\sum_{i=0}^3 \ket{i}\bra{i}\right]+t\ket{j}\bra{j}
\]
are separable states;  for $t<-1/3$, $\rho_t$ is no longer a state since it is no longer positive semi definite,
the eigenvalue of $\ket{j}$ becoming negative. Finally, if $\ket{j} = \bket{11}$, proposition \ref{prop-trom-zero}
tells us that there are
entangled states arbitrarily close to $\displaystyle\frac{1}{3}\sum_{i=0}^2 \ket{i}\bra{i}$.

\subsection{Example 2 -- The Thermal State of Two Qubits}\label{section:thermal}

The thermal state on two qubits is the state
\begin{align*}\quad
\rho_\Theta = \frac{(1+\eta)^2}{4} \bket{00}\bbra{00} &+ \frac{1-\eta^2}{4} \Big[\bket{01}\bbra{01} 
 + \bket{10}\bbra{10}\Big]\\& + \frac{(1-\eta)^2}{4}\bket{11}\bbra{11}.
\end{align*}
The state $\bket{11}$ is a $0$-eigenstate  of $\rho_p = (1+p)\rho_\Theta - p\bket{11}\bbra{11}$  if 
$(1-\eta)^2(p+1) = 4p$ and
a proposition similar to proposition \ref{prop-trom-zero} can be written for $\rho_p$. 
However, both cases of Sections \ref{section:pseudo-pure} and \ref{section:thermal} will be dealt with, by a generalization done in example~4.

The thermal state will get more attention later on, when we discuss $N$ qubits.

\subsection{Example 3 --- 2-Void State}

Example 3, using a 2-void state, is as follows:

\begin{proposition}\label{basicex}
In $\mathscr{H}_2\otimes\mathscr{H}_2$ there are entangled states arbitrarily close to the state $\rho =\displaystyle\frac{1}{2}\Big[\bket{01}\bbra{01} + \bket{10}\bbra{10}\Big]$.
\end{proposition}
\begin{proof}
Here again, $\bket{11}$ is an eigenstate of $\rho$ of $0$ eigenvalue.
Let $\rho_1 = \ket{\psi_+}\bra{\psi_+}$ with $\ket{\psi_+} = \frac{1}{\sqrt{2}}\big[\bket{01} + \bket{10}\big]$ and $\rho_\epsilon = (1-\epsilon)\rho + \epsilon \rho_1$. 
Then $\PT(\rho_\epsilon)$, where $\mathrm{T}$ is the transpose operator,
is 
\[\quad
\PT\begin{bmatrix} 0 & \ 0\  &\  0 \ & 0 \\
   0 & \ 1/2 \ & \ \epsilon/2 \ & 0\\
   0 & \epsilon/2 & 1/2 & 0\\
   0 & 0 & 0 &0\end{bmatrix}
   =
   \begin{bmatrix} 0 & \ 0\  &\  0 \ & \epsilon/2 \\
   0 & \ 1/2 \ & \ 0 \ & 0\\
   0 & 0 & 1/2 & 0\\
   \epsilon/2 & 0 & 0 &0\end{bmatrix}
\]
with characteristic equation $(\lambda -1/2)^2(\lambda^2 - \epsilon^2/4) = 0$ and eigenvalues $1/2$, $\epsilon/2$ and $-\epsilon/2$; by
the Peres criterion\footnote{Although the Peres Criterion is well known, it is provided
for completeness of the exposition in appendix~\ref{app:Peres}.}
 \citep{Peres96}, $\rho_\epsilon$ is thus entangled for all $1> \epsilon > 0$ and,
of course, 
$\delta(\rho,\rho_\epsilon) \leq \epsilon$.\qed
\end{proof}

In fact,  there was no need to solve the characteristic equation to show that $\PT(\rho_\epsilon)$ is not positive semi definite. That
can be seen directly from the matrix of $\PT(\rho_\epsilon)$ because there is a $0$ on the main diagonal for which the corresponding row and column
are not zero: 
This is a consequence of the following well known lemma with $\ket{\phi} = \bket{11}$ and $\ket{\psi}=\bket{00}$;
indeed, since by very definition of the partial transpose $\bra{i_1j_1}\PT(\rho_\epsilon)\ket{i_2j_2} = 
\bra{i_2j_1}\rho_\epsilon\ket{i_1j_2}$, it follows that
 $\bbra{11}\,\PT(\rho_\epsilon)\,\bket{11} = \bbra{11}\rho_\epsilon\bket{11}= 0$ and $\bbra{11}\,\PT(\rho_\epsilon)\,\bket{00} 
 = \bbra{01}\rho_\epsilon\bket{10} \neq 0$.

\begin{lemma}\label{negativeHermitian}
Let $A$ be a Hermitian operator on $\mathscr{H}$; if there are $\ket{\phi}$ and $\ket{\psi}$ such that $\bra{\phi}A\ket{\phi}=0$ and $\bra{\phi}A\ket{\psi}\neq 0$ then
$A$ is not positive semi definite.
\end{lemma}
\begin{proof}See  appendix \ref{appa}.
\end{proof}

\subsection{Example 4 --- A Generalization}

Example 4 generalizes examples 1, 2 and 3:

\begin{theorem}\label{propzero}
If the standard basis is the eigenbasis of a state $\rho$ on $\mathscr{H}_2\otimes\mathscr{H}_2$, and if the eigenvalue of 
$\bket{11}$ is $0$, then there are states arbitrarily close to $\rho$ that are entangled. The same holds for any
other eigenstate and any product  eigenbasis. 
\end{theorem}

\begin{proof}
Let indeed
\[\quad
\rho= \lambda_{\ttrm{00}} \,\ket{\texttt{00}}\bra{\texttt{00}} + \lambda_{\ttrm{01}}\, \ket{\texttt{01}}\bra{\texttt{01}} + \lambda_{\ttrm{10}}\,\ket{\texttt{10}}\bra{\texttt{10}},
\]
i.e. $\ket{\texttt{11}}$ has eigenvalue $\lambda_{11} = 0$. Let
\[\quad
\rho_1 = \rho_{\psi_+}=
\frac{1}{2}\Big[ \ket{\texttt{01}}\bra{\texttt{01}} + \ket{\texttt{01}}\bra{\texttt{10}} + \ket{\texttt{10}}\bra{\texttt{01}} + \ket{\texttt{10}}\bra{\texttt{10}}\Big]
\]
and $\rho_\epsilon = (1-\epsilon)\rho + \epsilon \rho_1$. 

The partial transpose $\PT(\ket{i_1j_1}\bra{i_2j_2})$ on basis states being equal to $\ket{i_2j_1}\bra{i_1j_2}$, it is
clear that $\PT(\rho) = \rho$.
% The matrix of $\rho$ being with real entries, its partial transpose with respect to the first system is $\rho$.
The partial transpose of $\rho_1$ is 
\[ % \begin{equation}\label{eq:parttransp}
\PT(\rho_1) = \frac{1}{2}\big[ \ket{\texttt{01}}\bra{\texttt{01}} + \ket{\texttt{11}}\bra{\texttt{00}} + \ket{\texttt{00}}\bra{\texttt{11}} + \ket{\texttt{10}}\bra{\texttt{10}}\big].
\] % \end{equation}
If follows that
\[\quad
\bbra{11}\,\PT(\rho_\epsilon)\,\bket{11} = 0, \quad \bbra{11}\,\PT(\rho_\epsilon)\,\bket{00} = \frac{\epsilon}{2}\,;
\]
by lemma \ref{negativeHermitian}, $\PT(\rho_\epsilon)$ is not positive semi definite if $\epsilon > 0$ and by the Peres criterion
it follows that the state $\rho_\epsilon$ is
then not separable;  since $\delta(\rho,\rho_\epsilon)\leq \epsilon$, there are  states arbitrarily close to $\rho$ that
are not separable.\qed
\end{proof}
Notice that all that is needed is that $\lambda_{11}=0$. Nothing prevents $\lambda_{10} = \lambda_{10} = 0$.
That implies, after a suitable choice of basis for the two systems, that any pure product state has 
arbitrarily close entangled states;
being two qubit states, they are also distillable~\citep{HorodeckiPRL97}, showing that there are arbitrarily close distillable
states.  
By symmetry, the result clearly holds if any of the other eigenvalues is known to be $0$ instead of $\lambda_{11}$.  Moreover the choice of product basis is arbitrary and the same argument applies for any product basis.

 \subsection{Example 5. A discordant product state with a non-product eigenbasis}\label{nonseparableeigen}

Take the discordant state 
\begin{equation}\quad\nonumber
\varrho = \frac{1}{2}\big[ \ket{\text{\texttt{00}}}\bra{\text{\texttt{00}}} \ + \ \ket{\text{\texttt{++}}}\bra{\text{\texttt{++}}}\big],
\end{equation}
where, this time, the representation is not spectral because $\ket{\text{\texttt{00}}}$,  $\ket{\text{\texttt{++}}}$, 
are not part of an orthonormal basis of $\mathcal{H}_2\otimes\mathcal{H}_2$; neither  $\ket{\text{\texttt{00}}}$ nor  $\ket{\text{\texttt{++}}}$
is an eigenvector of $\varrho$. The state $\varrho$ is equal to
$\displaystyle\frac{1}{2}\left[ \ket{\text{\texttt{00}}}\bra{\text{\texttt{00}}} \ + \ \frac{1}{4}\sum_{ijkl}\ket{ij}\bra{kl}\right]$
and is represented, in the standard basis, by the matrix
\[\quad\quad
\frac{1}{8}\begin{bmatrix} 5 &  \ 1\   & \ 1\   &  \ 1\\\
 1\  & \ 1\  & \ 1\  & \ 1  \\
 1\  & \ 1\ & \ 1\  & \ 1  \\
 1\  & \ 1\  & \ 1\  & \ 1 \end{bmatrix}
\]
The state $ \ket{\text{\texttt{++}}}$ being equal to $\frac{1}{2}[  \ket{\text{\texttt{00}}} +  \ket{\text{\texttt{01}}} +  \ket{\text{\texttt{10}}} +  \ket{\text{\texttt{11}}}$, it is easy to check that the states $\ket{\psi_0} = \displaystyle \frac{ \ket{\text{\texttt{00}}} +  \ket{\text{\texttt{++}}}}{\sqrt{3}}$
and $\ket{\psi_1} =  \ket{\text{\texttt{00}}} -  \ket{\text{\texttt{++}}}$ are normalized and orthogonal, and also that
\[\quad
\varrho = \frac{3}{4} \ket{\psi_0}\bra{\psi_0}\ +\ \frac{1}{4} \ket{\psi_1}\bra{\psi_1}.
\]
Since the spectral decomposition is unique, this shows that the spectral decomposition of $\varrho$ has eigenvectors that are not separable
even though $\varrho$ is itself separable: separability of a state does not imply that its eigenbasis is made out of separable states.

 As we will see in sections \ref{sec:QCs} and  \ref{sec:qutrit} the absence of a separable  eigenbasis is not a necessary and sufficient condition for discord. 

\begin{proposition}
If $\rho =\displaystyle 
\frac{1}{2}\big[ \ket{\text{\texttt{00}}}\bra{\text{\texttt{00}}} \ + \ \ket{\text{\texttt{++}}}\bra{\text{\texttt{++}}}\big]$ and $\rho_{\phi_+} = \ket{\phi_+}\bra{\phi_+}$ with $\ket{\phi_+} = \frac{1}{\sqrt{2}}[
\ket{\texttt{00}} + \ket{\texttt{11}}]$, then for $0 < \epsilon \leq 1$, the states $\rho_\epsilon = (1-\epsilon)\rho + \epsilon\rho_{\phi_+}$ are all entangled.
\end{proposition}
\begin{proof}
Clearly $\PT(\rho) = \rho$. Also
\[
\PT(\rho_{\phi_+}) = \frac{1}{2} \big[ \ket{\texttt{00}}\bra{\texttt{00}} + \ket{\texttt{10}}\bra{\texttt{01}} + 
\ket{\texttt{01}}\bra{\texttt{10}} + \ket{\texttt{11}}\bra{\texttt{11}} \big]
\]
Let $\ket{\psi_-} = \frac{1}{\sqrt{2}}[\ket{\texttt{01}}-\ket{\texttt{10}}]$. A simple calculation shows that
\[
\quad\bra{\psi_-}\PT(\rho_{\phi_+})\ket{\psi_-} = -\frac{1}{2}\quad\text{and}\quad
\bra{\psi_-}\rho\ket{\psi_-} = 0
\]
It follows that $\bra{\psi_-}\PT(\rho_\epsilon)\ket{\psi_-} = - \epsilon/2$ which shows directly that the partial transpose
of $\rho_\epsilon$
is not positive semi-definite for $\epsilon>0$. The state $\rho_\epsilon$ is consequently entangled
and the discordant state $\rho$ is boundary separable.\qed
\end{proof}

\subsection{Example 6 - A Classical-Quantum state}\label{sec:QCs}

As mentioned in the introduction (sec \ref{sec:set}) the definition of discord is asymmetric. The following state is classical with respect to $A$, but it is not
classical with respect to $B$; it  becomes discordant under the interchange of  the subsystems $A \leftrightarrow B$. Such states are often called 
classical-quantum \citep{DiscordRMP}.

\begin{proposition}\label{bb84basis}
Let
\begin{equation}\label{eq:CQs}
\rho= \lambda_{\ttrm{00}} \,\bket{{00}}\bbra{{00}} +  \lambda_{\ttrm{01}} \,\bket{{01}}\bbra{{01}} +
\lambda_{\ttrm{1+}}\, \bket{{1+}}\bbra{{1+}} + \lambda_{\ttrm{1-}}\,\bket{{1-}}\bbra{{1-}}
\end{equation}
If any of the eigenvalues is $0$, then there are states arbitrarily close to $\rho$ that are entangled.
\end{proposition}

\begin{proof}
This time we first prove if $\lambda_{00} = 0$ i.e.
if 
\[\quad
\rho= \lambda_{\ttrm{01}} \,\ket{\mathtt{01}}\bra{\mathtt{01}} + \lambda_{\ttrm{1+}}\, \bket{{1+}}\bbra{{1+}} + \lambda_{\ttrm{1-}}\,\bket{{1-}}\bbra{{1-}}.
\]
  Let  again 
    $\rho_1 = 
    \frac{1}{2}\big[\bket{01}\bbra{01} + \bket{01}\bbra{10} + 
    \bket{10}\bbra{01} + \bket{10}\bbra{10}\big]$ and $\rho_\epsilon = (1-\epsilon)\rho + \epsilon\rho_1$. Then
    $\bbra{00}\PT(\rho_\epsilon)\bket{00}  = 0$ and
    $\bbra{00}\PT(\rho_\epsilon)\bket{11} = \epsilon/2$ so that $\rho_\epsilon$ is not
    positive semi-definite by Lemma \ref{negativeHermitian} and $\rho_\epsilon$ is thus  entangled by 
    the Peres criterion. 
 Had we written explicitly the matrix, we would have seen the following pattern
     \[\quad
  \PT(\rho_\epsilon) \quad = \quad
\bordermatrix{ & \mathtt{00} & \mathtt{01} & \mathtt{10} & \mathtt{11} \cr
 \mathtt{00} & 0 &   &    &   \epsilon/2\cr
 \mathtt{01} &     &  &  & \cr
 \mathtt{10} &     &   &  \cr
\mathtt{11} & \epsilon/2    &   &   &  \cr}
   \]
with a $0$ entry on the main diagonal for which the line is not identically $0$ and concluded that $\PT(\rho_\epsilon)$ is
not positive semi definite if $\epsilon \neq 0$.

In this proof, it was assumed that $\lambda_{00} = 0$ but the same result
 holds if the eigenvalue of any other basis element is $0$; for instance, if the eigenvalue of $\bket{1-}$ is
$0$, then applying $X\otimes XH$ maps the basis onto itself and $\bket{1-}$ onto $\bket{00}$; for $\bket{1+}$ we need to apply $X\otimes H$,
and for $\ket{01}$ we apply $I\otimes X$.\qed
\end{proof}

The state in equation \eqref{eq:CQs}  has a product eigenbasis.  We can make it discordant by interchanging the subsystems $A \leftrightarrow B$ in which case it is a discordant state with a product eigenbasis. 

\subsection{\R{Two qubits - discussion}}
\R{ In this section we provided a number of examples of boundary separable states. More generally, we  showed (Theorem \ref{propzero}) that any two qubit state which has a product basis and is not full rank is on the boundary. We also showed  that separable states may have  eigenstates which are not separable (example 5) and that discordant states can have a separable eigenstates (example 6).  }

\section{Two qutrits}\label{sec:qutrit}
Some of the subtleties of bipartite systems cannot be seen in qubit-qubit pairs and  qubit-qutrit pairs. These include UPBs  and bound entangled states \citep{UPB} and locally indistinguishable product states \citep{NLWE}. 

\subsection{A mixture of locally indistinguishable product states}

Consider the following states on a bipartite system where $\ket{a\pm b}$ denotes the state $\frac{1}{\sqrt{2}}[\ket{a} \pm \ket{b}]$.
\begin{equation}\quad\quad
\begin{array}{lll}
\ket{\psi_1}=&|1\rangle&|1\rangle\\
\ket{\psi_2}=&|0\rangle&|0+1\rangle\\
\ket{\psi_3}=&|0\rangle&|0-1\rangle\\
\ket{\psi_4}=&|2\rangle&|1+2\rangle\\
\ket{\psi_5}=&|2\rangle&|1-2\rangle\\
\ket{\psi_6}=&|1+2\rangle&|0\rangle\\
\ket{\psi_7}=&|1-2\rangle&|0\rangle\\
\ket{\psi_8}=&|0+1\rangle&|2\rangle\\
\ket{\psi_9}=&|0-1\rangle&|2\rangle
\end{array}
\label{9states}
\end{equation}

\begin{proposition}\label{prop:9states}The state  \  $\displaystyle 
\rho_0 = 1/8 \sum_{i=2}^9 \ket{\psi_i}\bra{\psi_i}
$  is boundary separable.\end{proposition}

\begin{proof}
 Let $\ket{\Psi_+} = \frac{1}{\sqrt{2}}[\ket{01} +\ket{10}]$ and $\rho_1 = \ket{\Psi_+}\bra{\Psi_+}$.
Thus $\rho_1 = \frac{1}{2}\big[ \ket{01}\bra{01} + \ket{10}\bra{10} + \ket{01}\bra{10} + \ket{10}\bra{01}\big]$.
%The partial transpose of $\rho_1$ is
%\[
%\PT(\rho_1) = \frac{1}{2}\big[\ket{10}\bra{01} + \ket{10}\bra{10} + \ket{00}\bra{11} + \ket{11}\bra{00}\big]
%\]
Let $\rho_\epsilon = (1-\epsilon)\rho_0 + \epsilon \rho_1$. For $2\leq i \leq 9$, $\braket{01}{\psi_i} = 0$ or
$\braket{\psi_i}{10} = 0$ so that $\bra{01}\rho_0\ket{10} = 0$ and thus $\bra{01}\rho_\epsilon\ket{10} = \epsilon/2$.
Also $\bra{11}\rho_0\ket{11} = \bra{11}\rho_1\ket{11} = 0$ and so $\bra{11}\rho_\epsilon\ket{11} = 0$. It follows that
$\bra{11}\PT(\rho_\epsilon)\ket{11} = \bra{11}\rho_\epsilon\ket{11} = 0$ and 
$\bra{11}\PT(\rho_\epsilon)\ket{00} = \bra{01}\rho_\epsilon\ket{01} = \epsilon/2 \neq 0$ and thus, for $0<\epsilon <1$
$\rho_\epsilon$ is entangled. \qed
\end{proof} 
See appendix \ref{sec:9statesproof} for a matrix based argumentation.

The states in eq \eqref{9states} cannot be distinguished using local operations \citep{NLWE} a property sometimes called \emph{non-locality without entanglement}. In \citet{BTdemons} it was shown that given a set of states and a prior probability distribution, there is no relation between  discord in the resulting mixed state and this property. Similarly a mixture of these 9 states is generally discordant with a product eigenbasis.

% We did not even use the fact that the partial transpose of  $\rho_0$ is  $\rho_0$.

\section{Two Qudits (Quantum Digits)}
We now consider bipartite systems, with each part  of dimension at least two. 

\subsection{The ball of  separable states}
\begin{proposition}\label{prop:mixed} All states $\rho$  in the finite ball  $Tr|\rho-\frac{1}{d}\openone|<\frac{1}{d}$ are separable and not boundary separable.
\end{proposition}

\begin{proof}
We know from the  Gurvits-Barnum bound \citep{GB} that  all states $\rho$ with $\|\rho-\frac{1}{d}\openone\|_2<\frac{1}{d}$ are separable. Using the fact that  $\|\cdot\|_2 \leq Tr|\cdot|$ we can see that the maximally mixed state and any state close enough to it is not boundary separable. \qed
\end{proof}

\begin{corollary} There are separable states that are not boundary separable.\label{cor:notboundary} \end{corollary}

\subsection{A general class of  boundary separable states}

We now consider states of bipartite systems $\mathcal{H}_A\otimes\mathcal{H}_B$ for which
$\dim\mathcal{H}_A\geq 2$ and $\dim \mathcal{H}_B\geq 2$.

\begin{lemma}\label{fundlemma}
Let $\rho_0$ be a state such that $\bra{00}\rho_0\ket{00} = 0$ and $\bra{10}\rho_0\ket{01} = 0$. Let 
$\rho_1$ be such that $\bra{00}\rho_1\ket{00} = 0$ and $\bra{10}\rho_1\ket{01} = re^{i\phi}$ with $r>0$.
Then $\rho_\epsilon  = (1-\epsilon)\rho_0 + \epsilon\rho_1$ is  entangled and distillable for all $0 < \epsilon \leq 1$. 
\end{lemma}
\begin{proof} The conditions on $\rho_0$ and $\rho_1$ imply that
\begin{align}
\ \bra{00}\rho_\epsilon\ket{00} &= (1-\epsilon)\bra{00}\rho_0\ket{00} + \epsilon\bra{00}\rho_1\ket{00} = 0, \label{zeroval}\\
\bra{10}\rho_\epsilon\ket{01} &= (1-\epsilon)\bra{10}\rho_0\ket{01} + \epsilon\bra{10}\rho_1\ket{01} =\epsilon re^{i\phi}.
\label{nonzeroval}
\end{align}
From \eqref{zeroval} and \eqref{nonzeroval} it follows that 
$\bra{00}\PT(\rho_\epsilon)\ket{00} = 0$ and 
$\bra{00}\PT(\rho_\epsilon)\ket{11} = \epsilon re^{i\phi} \neq 0$ and consequently $\PT(\rho_\epsilon)$ is not positive semi-definite and $\rho_\epsilon$ is entangled.
We can do better.
Let
  $\ket{\Psi_\theta} = \sin(\theta)\ket{11}-e^{i\phi}\cos(\theta)\ket{00}$. Then, since $\bra{00}\rho_\epsilon\ket{00} = 0$, $\bra{01}\rho_1\ket{10} = \overline{\bra{10}\rho_1\ket{01}} = re^{-i\phi}$ and
  $\bra{\Psi_\theta} = \sin(\theta)\bra{11} -e ^{-i\phi}\bra{00}$, il follows that
\begin{align*}
\quad\bra{\Psi_\theta}&\PT(\rho_\epsilon)\ket{\Psi_\theta} \\
 =& \sin^2(\theta)\bra{11}\rho_\epsilon\ket{11} -2\sin(\theta)\cos(\theta)r\epsilon.
\end{align*}
If $\bra{11}\rho_\epsilon\ket{11}=0$, the result  is negative for $0 < \theta < \pi/2$. Else
for all $r\epsilon \neq 0$ there exits $\theta$ such that $\bra{11}\rho_\epsilon\ket{11} < 2\cot(\theta)r\epsilon$
 i.e.such that $\bra{\Psi_\theta}\PT(\rho_\epsilon)\ket{\Psi_\theta} < 0$,
which implies that $\rho_\epsilon$  is entangled and distillable by a lemma of \citet{PhysRevA.65.042327} (cf. appendix \ref{distlemma}). 

\end{proof}

We shall now  consider separable void states with a separable $0$ eigenvector. They are not only boundary separable
but independently of the dimensions (larger than 2), they have arbitrarily close distillable entangled states.
\begin{proposition}\label{prop:real}
Let  $\rho$ be a separable state of a bipartite system $\mathscr{H}_A\otimes\mathscr{H}_B$ ($\dim \mathcal{H}_A \geq 2$,
$\dim \mathcal{H}_B\geq 2$) that has a product state $\ket{\phi_0\psi_0}$ as eigenstate with $0$ eigenvalue.
%; let us assume that $\rho$ is represented by a real matrix;
% let us also assume that the state $\ket{\phi_1}$ has real coefficients.
Then $\rho$ is  boundary separable;  moreover there are entangled states arbitrarily close to $\rho$ that are distillable.
\end{proposition}
\begin{proof}
 We may   assume that $\ket{\phi_0} = \ket{0}_A$ and $\ket{\psi_0} = \ket{0}_B$ (cf Appendix \ref{app:Peres} and \ref{distlemma}) so that
$\ket{\phi_0\psi_0} = \ket{00}$ (dropping the indices $A$ and $B$, as was done till now) and perform the partial transpose
 using the basis $\ket{0}$, $\ket{1}$, etc, of $\mathcal{H}_A$. 
 Since $\rho$ is  separable, $\PT(\rho)$ is  a state and, from $\bra{00}\PT(\rho)\ket{00}=
\bra{00}\rho\ket{00} = 0$,
it follows that $\ket{00}$ is a $0$ eigenvector of $\PT(\rho)$ and thus $\bra{00}\PT(\rho)\ket{11} = 0$
i.e. $\bra{10}\rho\ket{01} = 0$.
 Let now $\rho_1$ be any state such that 
 $\bra{00}\rho_1\ket{00} = 0$ and $\bra{10}\rho_1\ket{01}\neq 0$.
 The conditions of Lemma \ref{fundlemma} are satisfied and $\rho_\epsilon = (1-\epsilon)\rho + \epsilon\rho_1$ is entangled 
 and distillable for all $0 < \epsilon \leq 1$.
\qed
\end{proof}

\begin{corollary}All pure product states are boundary separable.\label{Cor:pureBS}\end{corollary}

\subsection{Discordant states}

\begin{proposition}: If $\rho$ and $\tau$ are classical  with respect to the basis $\{\ket{i}\}$ then so is
the state $(1-t)\rho+t\tau$ for any valid $t$. \end{proposition}
\begin{proof}
That follows directly from Proposition \ref{DiscordCaract}. \qed
\end{proof}

\begin{theorem}\label{prop:dicordidenity}: The state $(1-t)\frac{\openone}{d}+t\rho$ is discordant if and only if $\rho$ is discordant,
where  $\openone$ is the identity of 
$\mathcal{H}_A\otimes\mathcal{H}_B$ and $d =\dim\left(\mathcal{H}_A\otimes\mathcal{H}_B\right)$.\end{theorem}

\begin{proof} This  follows directly from statement 4 in Proposition \ref{prop:discord}  and the fact that $\openone$ is invariant under any dephasing channel. \qed
\end{proof}

\begin{proposition}\label{prop:costdisc}Consider the 
state \[\quad\rho=\displaystyle\sum_{i,j} \mu_{ij}\ket{i}\bra{i}\otimes U_i\ket{j}\bra{j}U_i^\dagger\,,\] 
with the standard basis chosen such that the smallest eigenvalue is $\mu_{00}$ and $U_0=\mathrm{I}$.
Let $\rho_1$ be any state of the bipartite system s.t. $\bra{00}\rho_1\ket{00} = 0$ and $\bra{10}\rho_1\ket{01} \neq 0$; 
%(for instance $\rho' = \rho_+ = \frac{1}{d}\sum_{ij}\ket{ii}\bra{jj}$ for
% $d = \dim(\mathcal{H}_A\otimes\mathcal{H}_B)$); 
then the state
$\rho_\epsilon=(1-\epsilon)\rho+\epsilon\rho_1$ is a discordant state for all  $0<\epsilon\leq 1$. \label{prop:classotdisc}\end{proposition}

\begin{proof}
 Let
 $\rho_v=\sum_{i,j}\frac{1}{1-d\mu_{00}}(\mu_{ij}-\mu_{00})\ket{i}\bra{i}\otimes U_i\ket{j}\bra{j}U_i^\dagger$
 if $\rho \neq \frac{\openone}{d}$ (the $\mu_{ij}$ are not all equal),
  else let $\rho_v = \ket{00}\bra{00}$;
then $\rho = z\rho_v + (1-z)\frac{\openone}{d}$ for $z = 1-d\mu_{00}$, $0\leq z\leq 1$,
$\bra{00}\rho_v\ket{00} = 0$ and $\bra{10}\rho_v\ket{01} = 0$.
From proposition \ref{prop:dicordidenity}  we know that
$\rho_\epsilon$ is discordant if and only if the state $\hat\rho_{\epsilon,z} = k\rho_{\epsilon,z}$ obtained by normalizing $\rho_{\epsilon,z} = (1-\epsilon)z\rho_v + \epsilon\rho_1$ is discordant.
It holds that $\hat\rho_{\epsilon,z} = (1-\epsilon')\rho_v + \epsilon'\rho_1$ with $\epsilon' = k\epsilon$
($0 < \epsilon' \leq 1$).
By Lemma \ref{fundlemma}, 
$\hat\rho_{\epsilon,z}$ is entangled (and distillable), so that $\rho_\epsilon$ is discordant and $\rho$ is boundary
% discordant.\qed
classical.\qed
\end{proof}

\begin{corollary}\label{cor:BC}
% {\bf Corollary 1:} 
All classical   states are boundary classical.
\end{corollary} 

\R{We note that one could arrive at corollary \ref{cor:BC} by using the fact that the set of classical states is  nowhere dense \citep{classicalvolume}. }

Proposition \ref{prop:classotdisc} provides a method to construct $\epsilon$-discordant states. If a classical state $\rho$ (or any state $\rho\in \mathcalus{C}_A$) 
is not boundary separable the $\epsilon$-discordant state is also dissonant (i.e.it is separable). In general there is no direct relation between discord and boundary separable states. 

\begin{proposition} There are classical  states (and states in $\mathcalus{C}_A$) that are boundary separable and discordant states that are not boundary separable\end{proposition}
\begin{proof}

Most of the examples of boundary separable states  above are classical with respect to $A$. Moreover all  pure product states are also uncorrelated   and boundary separable.  

 From proposition \ref{prop:dicordidenity} and we know there are discordant states arbitrarily close to the maximally mixed states. From proposition  \ref{prop:mixed} we know that these states are not boundary separable. \qed
\end{proof}

\subsection{\R{Two qudits - discussion}}
\R{In this section we showed that there are separable states that are not boundary separable (corollary \ref{cor:notboundary}). We then presented a general class of boundary separable states and showed that in general all pure product states are boundary separable (Corollary \ref{Cor:pureBS}). Finally we discussed  depolarized  discordant states    (Theorem \ref{prop:dicordidenity}), showed that all classical states are boundary classical (Corollary \ref{cor:BC}) and that there is no direct relation between discord and boundary separable states.}

\section{Multiple qubits} \label{sec:multi}
\subsection{Extrapolated Pseudo-Pure States of \textit{N} Qubits}
Let us consider   states of the form
\[\quad
\rho_t = (1-t) \frac{\mathrm{I}}{2^N} + t \bket{11\ldots 1}\bbra{11\ldots 1}
\]
where $\mathrm{I}$ is the identity matrix, but this time of size $2^N\times 2^N$, and $t < 0$. With $(1-t_b) + 2^Nt_b = 0$ i.e. $t_b = -\frac{1}{2^N-1}$, 
$\rho_b = \rho_{t_b}$ becomes a $1$-void state, with $\bket{11\ldots 1}$ as $0$-eigenvector. The states $\rho_t$ for $t_b \leq t \leq 0$ are all
clearly separable; their matrix is diagonal in the standard basis, with non negative eigenvalues. Only the eigenvalue of $\bket{11\ldots 1}$ decreases.

\subsubsection{$\rho_b$ Is a Boundary Separable State.}
We choose arbitrarily the first bit and 
show that there are $\epsilon$ close entangled states for which the first qubit is entangled with the others. Let
 $\ket{\mathbf{1}} = \ket{\texttt{1}^{N-1}}$, i.e.  $N-1$ bits equal to one.
The eigenstate of $\rho_b$ with $0$ eigenvalue can be written as $\ket{\texttt{1}^N} = \ket{\texttt{1}}\ket{\mathbf{1}}$ and 
Proposition \ref{prop:real}  applies.

\subsubsection{Trace Distance Between \boldmath $\epsilon$-Entangled States and the Completely Mixed State.}
The trace distance between $\rho_b$ and $\mathrm{I}/2^N$ is
\[
\frac{1}{2}\tr \Big| (1-t_b)\frac{\mathrm{I}}{2^N} + t_b\ket{\texttt{1}^N}\bra{\texttt{1}^N} - \frac{\mathrm{I}}{2^N}\Big|
= \frac{|t_b|}{2}\tr\Big| \frac{\mathrm{I}}{2^N} - \ket{\texttt{1}^N}\bra{\texttt{1}^N}\Big|.
\]
The trace of $\big|\mathrm{I}/2^N - \ket{\texttt{1}^N}\bra{\texttt{1}^N}\big|$ is $(2^N-1)\times 1/2^N + 1 - 1/2^N = 2 - 2/2^N$.
The trace distance is thus
\[\quad
\delta(\frac{\mathrm{I}}{2^N}, \rho_b) = \frac{1}{2^N-1} \left(1-\frac{1}{2^N}\right) = \frac{1}{2^N}
\].
Conclusion: for any $\epsilon >0$ there are entangled states at distance at most $2^{-N}+\epsilon$ of the completely mixed state. Indeed, by the triangle inequality,
\[\quad
\delta(\frac{\mathrm{I}}{2^N},\rho_\epsilon) \leq \delta(\frac{\mathrm{I}}{2^N},\rho_b) + \delta(\rho_b,\rho_\epsilon) \leq 2^{-N} + \epsilon
\].
\subsection{The $N$ Qubit Thermal State}
The thermal state of one qubit is 
\[\quad
\rho_\Theta = \begin{bmatrix} \frac{1+\eta}{2} & 0 \\ 0 & \frac{1-\eta}{2}\end{bmatrix} = \frac{1+\eta}{2}\ \ket{0}\bra{0} + \frac{1-\eta}{2}\ \ket{1}\bra{1}.
\]
The thermal state of $N$ independent qubits (with the same $\eta$) is 
\begin{equation}\nonumber\quad
\rho^N_\Theta = \rho_\Theta^{\otimes N} =  \sum_{i\in\{0,1\}^N} \left(\frac{1+\eta}{2}\right)^{N-|i|}\left(\frac{1-\eta}{2}\right)^{|i|} \ket{i}\bra{i}.
\end{equation}
where $|i|$ is the Hamming weight of the string $i$, i.e. the number of bits equal to $1$ in $i$, each $1$ giving a minus sign, and each $0$
a plus sign.
The thermal state is not only separable but it has an eigenbasis  consisting of product states
The smallest eigenvalue is given by the eigenvector $\ket{i} = \ket{1^N}$, i.e. all qubits are $1$ and it is 
\[\quad
\lambda_{\ket{1^N}} = \left(\frac{1-\eta}{2}\right)^N
\]
which is exponentially small with $N$. 

\subsubsection{Extrapolated States Close to the Thermal State.}

Let us consider the extrapolated states
\[\quad
\varrho_t = (1-t)\rho_\Theta^N + t \ket{1^N}\bra{1^N}
\]
for $t<0$ ($t = -p$ for some positive real number $p$). They are all separable and when the eigenvalue of $\ket{1^N}\bra{1^N}$ becomes $0$, $\varrho_t$ is a void state.
That happens when $(1-t)\big[(1-\eta)/2\big]^N + t = 0$ i.e
\[\quad
 t_b= - \frac{\lambda_{\ket{1^N}}}{1-\lambda_{\ket{1^N}}} = -\lambda_{\ket{1^N}} - \lambda_{\ket{1^N}}^2 - \ldots
\]
a very small value, equal to $-\lambda_{\ket{1^N}} = -((1-\eta)/2)^N$ if we neglect terms of higher order. 
The trace distance between $\varrho_b$ and $\rho_\Theta^N$ is
\begin{align*}\quad
\delta(\varrho_b,\rho_\Theta^N) &= \frac{1}{2}\tr\Big|(1-t_b)\rho_\Theta^N + t_b\ket{1^N}\bra{1^N} - \rho_\Theta^N\Big|
\\&= \frac{|t_b|}{2} \tr\Big|\rho_\Theta^N - \ket{1^N}\bra{1^N}\Big|.
\end{align*}
The eigenvectors of $\rho_\Theta^N - \ket{1^N}\bra{1^N}$ are those of $\rho_\Theta^N$ and the eigenvalues are left unchanged except
for the eigenvector $\ket{1^N}$ whose eigenvalue of $\lambda_{\ket{1^N}}$ is decreased by $1$ which implies that the sum of the
absolute values of the eigenvalues is increased by $1 - \lambda_{\ket{1^N}}$ and
\begin{align*}\quad
\delta(\rho_\Theta^N, \varrho_b) &= \frac{|t_b|}{2}\left( 2 - \lambda_{\ket{1^N}}\right)
\\&= \frac{1}{2} \frac{\lambda_{\ket{1^N}}}{\displaystyle 1- \lambda_{\ket{1^N}}}\left( 2 - \lambda_{\ket{1^N}}\right) 
\\&= \frac{1}{2}\left(\lambda_{\ket{1^N}} + \frac{\lambda_{\ket{1^N}}}{1-\lambda_{\ket{1^N}}}\right) \\
& =\lambda_{\ket{1^N}} + \frac{1}{2}\lambda_{\ket{1^N}}^2 + \frac{1}{2}\lambda_{\ket{1^N}}^3 \ldots
\end{align*}
which is $\lambda_{\ket{1^N}}$ if we neglect terms of higher order. That distance is  exponentially small with $N$.

\subsubsection{\boldmath$\varrho_b$ Is a Boundary Separable State.} 

We now show that there are entangled states arbitrarily close to $\varrho_b$. 
We choose again arbitrarily the first bit and 
show that there are $\epsilon$ close entangled states for which the first qubit is entangled with the others. Let
 $\ket{\mathbf{1}} = \ket{\texttt{1}^{N-1}}$, i.e.  $N-1$ bits equal to one, and let $\ket{v}$ be any $N-1$ bit string with at least one 
bit equal to zero. The eigenstate of $\varrho_b$ with $0$ eigenvalue is $\ket{\texttt{1}^N} = \ket{\texttt{1}}\ket{\mathbf{1}}$. 
% That state has real entries when expressed in the standard basis. So does the state $\varrho_b$. 
Proposition \ref{prop:real}  applies again.

\subsubsection{Entangled States Close to the Thermal State}

We have just proven that or any $\epsilon > 0$, there are
entangled states $\varrho_\epsilon$ such that $\delta(\varrho_b,\varrho_\epsilon) \leq \epsilon$.
By the triangle inequality (since the trace distance is a distance in the sense of metric spaces),
the distance between  those states $\varrho_\epsilon$ and  $\rho_\Theta^N$ is such that
\[\quad
\delta(\rho_\Theta^N,\varrho_\epsilon) \leq \delta(\rho_\Theta^N,\varrho_b) + \delta(\varrho_b,\varrho_\epsilon) \leq \delta(\rho_\Theta^N,\varrho_b) + \epsilon
\]
which implies that for any $\epsilon >0$ there are entangled states in a ball of trace-distance radius 
\[\quad
\epsilon + \left(\frac{1-\eta}{2}\right)^N + \frac{1}{2}\left(\frac{1-\eta}{2}\right)^{2N}  + \frac{1}{2}\left(\frac{1-\eta}{2}\right)^{3N} \ldots
\]
around the thermal state $\rho_\Theta^N$ of $N$ qubits where $\displaystyle\left(\frac{1-\eta}{2}\right)^N = \lambda_{\ket{1^N}}$ is exponentially small in $N$.

\section{Discussion}
We used extrapolation and interpolation to study the boundaries of some subsets of states and to make some connections between different notions of entanglement and quantum correlations. The majority of our results concern boundary separable states. We showed various classes of these states that play a significant role in quantum computing. In particular the classically correlated states in the computational basis and thermal states. 

Our results are related to results on robustness against various types of noise. States near the boundary are generally more fragile than those far away from it.  It is then interesting to note that although thermal states are not entangled they can become entangled by small fluctuation in the right direction, moreover the entanglement is distillable. 

\subsection{Discord and entanglement}
While discord and entanglement are very  different from an operational perspective \citep{DiscordRMP,ABPRA}, pure state entanglement and discord (for all states) share many similar mathematical properties \citep{DiscordRMP}. Many of these appeared in the work above, in particular in the property of boundaries. All classical states 
 are boundary classical (Corollary \ref{cor:BC}), similarly  all pure product are at the boundary of the set of entangled states (Corollary \ref{Cor:pureBS}). This opens a number of interesting questions regarding operations on pure and mixed states. For example we showed in proposition \ref{prop:classotdisc} 
that mixing any classical state (and any classical state with respect to $A$) with an entangled state will make it discordant, similarly mixing a pure product state with an  entangled state will make it entangled (Corollary \ref{Cor:pureBS}). However the former is not an extension of the latter into mixed states  since both  take pure states to a mixed states. If one considers unitary operations on pure states there is come discrepancy. Universal entanglers (unitary operations that entangle all pure states) are known to exits only in higher dimensions \citep{UE}. 

\subsection{\R{Quantum computing}}
\R{It is currently an open question whether it is possible to efficiently simulate all quantum computations that produce (or consume) no entanglement \citep{BBM2017}. Surprisingly, it is not even clear if it is possible to simulate all quantum computation that produce (or consume) no discord. Boundary (classical or separable) states may play a critical role in these types of simulations since even small errors can  cause the states to become discordant or entangled. This issue was pointed out for the case of  discord free (concordant) computation in \citep{CableBrowne} where the entire computation happens on the boundary (see also \citep{DattaShaji}). }

\R{A second issue  is related to the entangling power of mixed state quantum computers, where the initial state is fixed. Here we showed that the thermal states used in various mixed state models can become entangled after a small perturbation.  However, we did not discuss the physical mechanism for such perturbations.  In a follow up paper \citep{BBM2017} we explore perturbations due to unitary operations that are $\epsilon$ close to the identity.  A question for future research involves the possible use of thermal states and some non-local operations to distill entanglement.  This question is especially important in the setting of NMR quantum computing. }

\section*{Acknowledgments}
MB was partly supported by NSERC and FCAR through INTRIQ.
AB was partly supported by NSERC, Industry Canada and CIFAR.
 TM was partly supported by the Israeli MOD.   AB and TM were partly supported The Gerald Schwartz \& Heather Reisman Foundation. AB is currently at the Center for Quantum Information and Quantum Control at the University of Toronto.

\newcommand{\doi}[1]{doi: \href{http://dx.doi.org/#1}{#1}}
\newcommand{\eprint}[1]{\href{http://arxiv.org/abs/#1}{#1}}

\bibliographystyle{spbasic}

\appendix

\bigskip
\section*{Appendix}

\section{The Peres Entanglement Criterion}\label{app:Peres}
Here are a few relevant remarks using the notations of the main article.
\subsection{Transpose and partial transpose}

Given a Hilbert space $\mathcal{H}$ and a basis $\big\{\ket{i}\big\}$ (we always assume finite dimensional systems), the transpose is defined by linearity on basis operators $\ket{i_1}\bra{i_2}$ by $\mathrm{T}(\ket{i_1}\bra{i_2}) = \ket{i_2}\bra{i_1}$.
It follows that for any linear operator $L$, 
$\bra{i_1}\mathrm{T}(L)\ket{i_2} = \bra{i_2}\,L\,\ket{i_1}$.
If $\rho$ is a state and $\rho = \sum_i \lambda_i \ket{\phi_i}\bra{\phi_i}$, one can check that
$\mathrm{T}(\rho) = \sum_i \lambda_i \ket{\overline{\phi}_i}\bra{\overline{\phi}_i}$ where
$\ket{\overline{\phi}}  = \sum_i \overline{a}_i \ket{i}$ if $\ket{\phi} = \sum_i a_i\ket{i}$,  $\overline{a}_i$ 
being the complex conjugate of $a_i$. It follows that $\mathrm{T}(\rho)$ is also a state, with same eigenvalues as $\rho$.

Given a compound system described by $\mathcal{H}_A\otimes\mathcal{H}_B$, the partial transpose with respect to the $A$
system is simply the operator $\PT$ i.e.
\[
\quad \PT\Big(\ket{i_1}\bra{i_2}\otimes \ket{j_1}\bra{j_2}\Big) = \ket{i_2}\bra{i_1}\otimes \ket{j_1}\bra{j_2}
\]
on basis elements. It also follows that for any operator $L$ on $\mathcal{H}_A\otimes\mathcal{H}_B$
\[
\quad \bra{i_1j_1}\PT(L)\ket{i_2j_2} = \bra{i_2j_1}\,L\,\ket{i_1j_2}
\]
\subsection{The Peres Criterion}
A state $\rho$ of a bipartite system $\mathcal{H}_A\otimes\mathcal{H}_B$ is said to be \emph{separable} if it
can be written in the form
\begin{align}
\quad\rho = \sum_i p_i\, \rho^A_i \otimes \rho^B_i &&p_i \geq 0,\ \sum_i p_i = 1 \label{def:separable}
\end{align}
where the $\rho^A_i$ (resp. $\rho^B_i$) are states of $\mathcal{H}_A$ (resp. $\mathcal{H}_B$);
if $\rho$ is not separable, it is
said to be \emph{entangled}.
If $\rho$ is given by \eqref{def:separable}, then
\[
\quad \PT(\rho) = \sum_i p_i \mathrm{T}(\rho^A_i)\otimes \rho^B_i
\]
and since the $\mathrm{T}(\rho^A_i)$ are states, this implies that $\PT(\rho)$ is itself a state (and separable). 
This implies in turn that $\PT(\rho)$ must be positive semi-definite. As a consequence, if $\PT(\rho)$ is
not positive semi-definite, then $\rho$ is not separable, i.e. it is entangled.
That is the statement of the Peres Criterion of entanglement \citep{Peres96}.

\subsection{Checking for positivity}
An operator $P$ is positive semi-definite if it is Hermitian and if for all pure states $\ket{\phi}$, 
$\bra{\phi} P\ket{\phi} \geq 0$ (iff $P$ has no negative eigenvalue).
For any state $\rho$ of $\mathcal{H}_A\otimes\mathcal{H}_B$, $\PT(\rho)$ is always
Hermitian. 
%When checking for positivity, we are using the fact that $\PT(\rho)$ is Hermitian
%if $\rho$ is a state. In fact $\PT(H)$ is Hermitian as soon as $H$ is.
%That follows from elementary calculation: if $H = \sum_{ij} \ket{i}\bra{j}\otimes H_{ij}$, 
%then $H^\dagger = \sum_{ij} \ket{j}\bra{i}\otimes H_{ij}^\dagger = \sum_{ij} \ket{i}\bra{j}\otimes H_{ji}^\dagger$ and  $H$ is
%Hermitian if and only if $H_{ij}^\dagger = H_{ji}$. It then follows  that $\PT(H)^\dagger = \left(\sum_{ij} \ket{j}\bra{i}\otimes H_{ij}\right)^\dagger = \sum_{ij}\ket{i}\bra{j}\otimes H_{ij}^\dagger = \PT(H)$.
To prove that it is not positive semi-definite, we need only find a $\ket{\Psi}$ such that
$\bra{\Psi}\PT(\rho)\ket{\Psi} < 0$. The partial transpose however depends on the basis chosen for $\mathcal{H}_A$.
We now show (using our notations) that whether $\PT(\rho)$ is positive semi-definite or not does not depend on the choice of that basis.
Indeed, let $\ket{e_i}$ be any orthonormal basis of $\mathcal{H}_A$. Then $\rho$ can always be written (in a unique way)
as $\rho = \sum_{ij} \ket{e_i}\bra{e_j}\otimes \rho_{ij}$ where the $\rho_{ij}$ are operators of $\mathcal{H}_B$.
Let $\mathrm{T}_e$ be the transpose operator in the basis $e$ i.e. $\mathrm{T}_e(\ket{e_i}\bra{e_j}) = \ket{e_j}\bra{e_i}$.
Then 
\begin{align*}
(\mathrm{T}_e\otimes I)(\rho) &= \sum_{ij} \mathrm{T}_e(\ket{e_i}\bra{e_j})\otimes \rho_{ij} = \sum_{ij} \ket{e_j}\bra{e_i}\otimes \rho_{ij}\\
\PT(\rho) &= \sum_{ij} T(\ket{e_i}\bra{e_j})\otimes \rho_{ij} = \sum_{ij} \ket{\overline{e}_j}\bra{\overline{e}_i}\otimes \rho_{ij}
\end{align*}
The $\ket{\overline{e}_i}$  also form an orthonormal basis of $\mathcal{H}_A$.
Now let $\ket{\Psi_e} = \sum_i \ket{e_i}\ket{\psi_i}$ be any pure state of $\mathcal{H}_A\otimes\mathcal{H}_B$.
Then $\ket{\Psi} = \sum_i \ket{\overline{e}_i}\ket{\psi_i}$ is also a pure state and
\[
\quad\bra{\Psi} \PT(\rho)\ket{\Psi} = \bra{\Psi_e}(\mathrm{T}_e\otimes I)(\rho)\ket{\Psi_e} = \sum_{ij} \bra{\psi_j}\rho_{ij}\ket{\psi_i}
\]

\subsection{Proof of Lemma \ref{negativeHermitian}}\label{appa}
\begin{proof}
Let us assume $P$ is positive semidefinite:
$P=\sum_i \lambda_i \ket{\phi_i}\bra{\phi_i}$ with $\lambda_i\geq 0$. 
If $\bra{\phi}P\ket{\phi} = 0$, then $\sum_i \lambda_i |\braket{\phi}{\phi_i}|^2 = 0$ and 
$\lambda_i\braket{\phi}{\phi_i} =0$ for all $i$
and thus  $\bra{\phi}P\ket{\psi} = \sum_i \lambda_i \braket{\phi}{\phi_i}\braket{\phi_i}{\psi} = 0$ for 
all $\ket{\psi}$.\qed
\end{proof}

\section{Distillability}\label{distlemma}
Note that the Peres Criterion is not a characterization. If the partial transpose of $\rho$ is positive semi-definite,
$\rho$ may still be entangled.
 Furthermore, if a state $\rho_\textit{ppt-ent}$ 
is entangled and 
admits a positive partial transpose then it is not distillable
(namely, one canot distill a singlet state out of many copies of
$\rho_\mathit{ppt-ent}$ via local operations and classical communication).
Such states are said to have ``bound entanglement''.
A characterization of distillable states can be found in \citet{Horodecki1997333}.
Here is the lemma as we use it, as stated in \citet{PhysRevA.65.042327}.

\begin{lemma}[\citet{PhysRevA.65.042327,Horodecki1997333}]
\quad A state $\rho$ of $\mathcal{H}_A\otimes\mathcal{H}_B$ is distillable if and only if there exists a positive integer $N$ and
a state $\ket{\Psi} = \ket{e_1f_1} + \ket{e_2f_2}$ such that
\[
\quad\bra{\Psi} \,\PT(\rho^{\otimes N})\,\ket{\Psi} < 0,
\]
where $\{e_1,e_2\}$ (resp. $\{f_1,f_2\}$) are two unnormalized orthogonal vectors of $\mathcal{H}_A^{\otimes N}$
(resp. $\mathcal{H}_B^{\otimes N}$).
\end{lemma}
%\section{Horodecki's Distillability Criterion}\label{app:Horodecki}
%\begin{theorem}
%An arbitrary  state $\rho$ of $\mathscr{H}_A\otimes\mathscr{H}_B$ is distillable if and only if there exists $n$ and projectors
%$P_A:\mathscr{H}_A^{\otimes n}\to\mathscr{H}_2$ and $P_B:\mathscr{H}_B^n\to\mathscr{H}_2$ (where $\mathscr{H}_2$ denotes a Hilbert space of dimension 2) such that if  $\hat{\rho}'$ is the state obtained by normalizing the operator
%\[
%\rho' = (P_A\otimes P_B)\rho^{\otimes n}(P_A\otimes P_B)
%\]
%of the system $\mathscr{H}_2\otimes\mathscr{H}_2$, then $\hat{\rho}'$  is entangled.
%
%\end{theorem}
%Notice that it was proven in~\citet{Horodecki19961} 
%that  a state $\hat{\rho}'$ of $\mathscr{H}_2\otimes\mathscr{H}_2$ 
%is entangled if and only if it has a partial transpose
%that is not positive semi definite 
%(the Peres criterion is then a characterization of entanglement
%as well as of distillability for bipartite states of two qubits).

\section{Proof of Proposition \ref{prop:9states} using matrices\label{sec:9statesproof}}

When the states $\ket{ij}$ are put in lexicographic order,  the partial transpose $\PT$ corresponds to transposing blocks in the block matrix, whereas $(\mathrm{I}\otimes\mathrm{T})$ corresponds to transposing each of the blocks individually.
The matrix of Proposition \ref{prop:9states} is a $3\times 3$ block matrix, with $3\times 3$ blocks.

We first calculate  for both $\rho_0$ and $\rho_1$ the entries  $(11,11)$ and $(01,10)$ (row $01$, column $10$ of their matrix). Those are $\bra{11}\rho_0\ket{11}=0$, and $\bra{01}\rho_0\ket{10}=0$ for $\rho_0$ and 
$\bra{11}\rho_1\ket{11}=0$ and $\bra{10}\rho_1\ket{10}=1/2$ for $\rho_1$. Those values were obtained in the main text.
The matrices for $\rho_0$ and $\rho_1$ are then the following (useless entries being kept blank).
\[\hfill
\begin{blockarray}{cccccccccc}
 & {00}  & {01}  & {02}  & {10}  & {11}  & {12}  & {20}  & {21}  & {22} \\
 \begin{block}{l@{\quad}(ccc|ccc|ccc)}
{00} &   &  &  &  &  &  &  &  &  \\
{01} &   &  &  &  0 &  &  &  &  &  \\
{02} &  &  &  &  &  &  &  &  &  \\ \cline{2-10}
{10} &  &   &  &  &  &  &  &  &  \\
{11} &  &  &  &  & 0 &  &  &  &  \\
{12} &  &  &  &  &  &  &  &  &  \\ \cline{2-10}
{20} &  &  &  &  &  &  &  &  &  \\
{21} &  &  &  &  &  &  &  &  &  \\
{22} &  &  &  &  &  &  &  &  &  \\
\end{block}
\end{blockarray}\hfill
\]
\[\hfill
\begin{blockarray}{cccccccccc}
 & {00}  & {01}  & {02}  & {10}  & {11}  & {12}  & {20}  & {21}  & {22} \\
 \begin{block}{l@{\quad}(ccc|ccc|ccc)}
{00} &   &  &  &  &  &  &  &  &  \\
{01} &   &  &  & 1/2 &  &  &  &  &  \\
{02} &  &  &  &  &  &  &  &  &  \\ \cline{2-10}
{10} &  &  &  &  &  &  &  &  &  \\
{11} &  &  &  &  & 0 &  &  &  &  \\
{12} &  &  &  &  &  &  &  &  &  \\ \cline{2-10}
{20} &  &  &  &  &  &  &  &  &  \\
{21} &  &  &  &  &  &  &  &  &  \\
{22} &  &  &  &  &  &  &  &  &  \\
\end{block}
\end{blockarray}
\hfill\]
Then, the $3\times 3$ block matrix is transposed, giving respectively for $\PT(\rho_0)$ and $\PT(\rho_1)$ 
the matrices:
\[\hfill
\begin{blockarray}{cccccccccc}
 & {00}  & {01}  & {02}  & {10}  & {11}  & {12}  & {20}  & {21}  & {22} \\
 \begin{block}{l@{\quad}(ccc|ccc|ccc)}
{00} &   &  &  &  &   &  &  &  &  \\
{01} &     &  &  &  &  &  &  &  &  \\
{02} &  &  &  &  &  &  &  &  &  \\ \cline{2-10}
{10} &  &  &  &  &  &  &  &  &  \\
{11} & 0 &  &  &  & 0 &  &  &  &  \\
{12} &  &  &  &  &  &  &  &  &  \\ \cline{2-10}
{20} &  &  &  &  &  &  &  &  &  \\
{21} &  &  &  &  &  &  &  &  &  \\
{22} &  &  &  &  &  &  &  &  &  \\
\end{block}
\end{blockarray}
\hfill\]
\[\hfill
\begin{blockarray}{cccccccccc}
 & {00}  & {01}  & {02}  & {10}  & {11}  & {12}  & {20}  & {21}  & {22} \\
 \begin{block}{l@{\quad}(ccc|ccc|ccc)}
{00} &  &  &  &  &   &  &  &  &  \\
{01} &     &  &  &  &  &  &  &  &  \\
{02} &  &  &  &  &  &  &  &  &  \\ \cline{2-10}
{10} &  &  &  &  &  &  &  &  &  \\
{11} & 1/2 &  &  &  & 0 &  &  &  &  \\
{12} &  &  &  &  &  &  &  &  &  \\ \cline{2-10}
{20} &  &  &  &  &  &  &  &  &  \\
{21} &  &  &  &  &  &  &  &  &  \\
{22} &  &  &  &  &  &  &  &  &  \\
\end{block}
\end{blockarray}
\hfill\]
The matrix of $\PT(\rho_\epsilon) = (1-\epsilon)\PT(\rho_0) + \epsilon \PT(\rho_1)$ is then
\[\hfill
\begin{blockarray}{cccccccccc}
 & {00}  & {01}  & {02}  & {10}  & {11}  & {12}  & {20}  & {21}  & {22} \\
 \begin{block}{l@{\quad}(ccc|ccc|ccc)}
{00} &  &  &  &  &   &  &  &  &  \\
{01} &     &  &  &  &  &  &  &  &  \\
{02} &  &  &  &  &  &  &  &  &  \\ \cline{2-10}
{10} &  &  &  &  &  &  &  &  &  \\
{11} & \epsilon/2 &  &  &  & 0 &  &  &  &  \\
{12} &  &  &  &  &  &  &  &  &  \\ \cline{2-10}
{20} &  &  &  &  &  &  &  &  &  \\
{21} &  &  &  &  &  &  &  &  &  \\
{22} &  &  &  &  &  &  &  &  &  \\
\end{block}
\end{blockarray}
\hfill\]
We see clearly that the matrix of $\PT(\rho_\epsilon)$ has a $0$ diagonal entry for which there is a non zero
entry on the corresponding row (or corresponding column). That implies that the matrix is not positive
semi-definite and consequently that $\rho_\epsilon$ is entangled.

Of course, the blank values in the density operator for $\rho_1$ could take any value without affecting the result;
in fact any density operator $\rho_1$ such that $\bra{11}\rho_1\ket{11} = 0$ and $\bra{01}\rho_1\ket{10} \neq 0$
could have been used instead to give entangled states that arbitrarily close to $\rho_0$.

\end{document}